\pdfoutput=1                     
\documentclass[a4paper,11pt]{article}

\usepackage[utf8]{inputenc}
\usepackage[T1]{fontenc}
\usepackage[english]{babel}
\usepackage{amsmath,amssymb,amsthm}
\usepackage{microtype}   
\usepackage[margin=1in]{geometry}
\usepackage[colorlinks=true,linkcolor=blue,citecolor=blue,urlcolor=blue]{hyperref}

\newtheorem{theorem}{Theorem}[section]
\newtheorem{lemma}[theorem]{Lemma}
\newtheorem{proposition}[theorem]{Proposition}
\newtheorem{corollary}[theorem]{Corollary}
\newtheorem{question}[theorem]{Question}
\theoremstyle{definition}
\newtheorem{definition}[theorem]{Definition}
\newtheorem{example}[theorem]{Example}
\theoremstyle{remark}
\newtheorem{remark}[theorem]{Remark}

\newcommand{\F}{\mathbb{F}}
\newcommand{\Z}{\mathbb{Z}}
\newcommand{\A}{A}
\newcommand{\wt}{\operatorname{wt}}
\newcommand{\tr}{\operatorname{tr}}
\newcommand{\Sp}{\operatorname{Sp}}
\newcommand{\SL}{\operatorname{SL}}
\newcommand{\GL}{\operatorname{GL}}
\newcommand{\Hom}{\operatorname{Hom}}
\newcommand{\Aut}{\operatorname{Aut}}
\newcommand{\ket}[1]{|#1\rangle}

\title{Minimal Counterexamples of the MacWilliams Extension \\Theorem
       for Stabilizer Codes}

\author{%
  Ali Assem Mahmoud\\[4pt]
  \normalsize Digital Technologies Research Centre, National Research Council of Canada,\\
  \normalsize 1200 Montreal Road, Ottawa, ON K1A 0R6, Canada\\[2pt]
  \normalsize Institute for Quantum Computing, University of Waterloo,\\
  \normalsize 200 University Avenue West, Waterloo, ON N2L 3G1, Canada\\[4pt]
  \normalsize\texttt{ali.mahmoud@uwaterloo.ca}\quad\textbullet\quad
  \normalsize\texttt{ali.mahmoud@nrc-cnrc.gc.ca}%
}

\date{\today}

\begin{document}

\maketitle

\begin{abstract}

The MacWilliams extension theorem fails for module alphabets with non-cyclic
socle, and the label alphabet of qudit stabilizer codes, $\F_{q^2}$ over
$\F_q$, is such an alphabet. Quantum error correction, however, only ever sees
\emph{self-orthogonal} additive codes, and whether that rigidity rescues the
theorem---equivalently, whether every weight-preserving isomorphism of
stabilizer groups is implemented by local Cliffords and a qudit
permutation---was asked by Gluesing-Luerssen and Pllaha and answered
negatively by Pllaha for particular qubit codes. We develop the negative
answer systematically and at the smallest possible scales. For every prime
power $q$ we construct a pair of $[[q+1,q-1]]_q$ stabilizer codes and a
weight-preserving isomorphism between them extending to no monomial
transformation; the codespaces are inequivalent even under arbitrary local
unitaries combined with permutations, though they share Shor--Laflamme
enumerators. Self-orthogonality is automatic here, by two elementary lemmas
which also show that Dyshko's threshold-length counterexamples were already
self-orthogonal, unremarked. For qubits we prove by exhaustive search that
length $3$ is minimal and the counterexample essentially unique. Dropping the
``idle qudit'' invariant that detects these, we find the minimal full-support
lengths: $4$ for a non-extendable isometry, $5$ for a weight-isometric pair
that is not monomially equivalent, realized by explicit $[[5,2]]$ codes; at
length $6$ all nontrivial stabilizer elements can have weight $\ge 4$. Whether
these codespaces are locally unitarily equivalent is posed as an open problem,
connecting the extension problem to the LU--LC circle of questions.
\end{abstract}



\section{Introduction}

\subsection{The classical extension theorem}

Two linear codes $C,D\le\F_q^n$ are \emph{monomially equivalent} if some permutation of coordinates followed by multiplication of each coordinate by a nonzero scalar maps $C$ onto $D$. MacWilliams proved in her thesis \cite{MacWilliams62} that monomial maps are the only source of Hamming-weight-preserving linear isomorphisms between linear codes over a finite field: every linear bijection $f\colon C\to D$ with $\wt(f(c))=\wt(c)$ for all $c\in C$ extends to a monomial transformation of $\F_q^n$. This \emph{extension theorem} is a cornerstone of the theory of code equivalence, complementing the (logically independent) MacWilliams identities relating the weight enumerators of a code and its dual.

A long line of work asked for which alphabets the extension theorem survives when $\F_q$ is replaced by a finite ring or, more generally, by a finite module $\A$ over a finite ring $R$ (``codes over $\A$'' are $R$-submodules of $\A^n$; the allowed monomial maps use $\Aut_R(\A)$ in each coordinate). Wood proved that finite Frobenius rings have the extension property and, conversely, that the extension property for Hamming weight characterizes them \cite{Wood99,Wood08}; for module alphabets the definitive statement is that $\A$ has the extension property with respect to Hamming weight if and only if $\A$ is pseudo-injective with cyclic socle \cite{Wood09,DinhLopez04}. Refinements to symmetrized weight compositions and other weights, and to the socle condition, appear among other places in \cite{GreferathSchmidt00,AssemJAA,Assem15,Dyshko16,DyshkoMDS,DyshkoGeom}. In particular, \cite{AssemJAA} (see also the M.Sc.\ thesis \cite{Assem15}) proves an extension theorem for symmetrized weight compositions for modules with a suitable socle condition, again with cyclicity of the socle as the decisive hypothesis, and works out the instructive example of a field pair $K\subsetneq L$: the alphabet $L$, viewed as a $K$-module, never has the extension property, the smallest instance being $K=\F_2\subset L=\F_4$.

That smallest instance is not a curiosity. Additive codes over $\F_4$---$\F_2$-linear subspaces of $\F_4^n$---are exactly the classical shadows of qubit stabilizer codes in the Calderbank--Rains--Shor--Sloane (CRSS) correspondence \cite{CRSS98}, and the same statement holds for qudits with $\F_{q^2}$ (equivalently the symplectic label space $\F_q^2$) in place of $\F_4$ \cite{Ketkar06}. The alphabet of quantum error correction is therefore an alphabet whose socle is as non-cyclic as possible: $\F_q^2$ is semisimple, equal to its own socle, and needs two generators. Wood's characterization thus already tells us that additive isometries need not extend, and explicit unextendable isometries of additive codes, including a family at the threshold length $q+1$, were constructed by Dyshko \cite{Dyshko16,DyshkoGeom}.

\subsection{The quantum question}\label{sec:quantumquestion}

Quantum error correction, however, never sees arbitrary additive codes. A stabilizer code is described by a \emph{self-orthogonal} additive code (with respect to the trace-Hermitian, equivalently trace-symplectic, form): this is the commutativity of the stabilizer group. The natural question is therefore:

\begin{quote}
Does the MacWilliams extension theorem hold for stabilizer codes? That is, does every weight-preserving isomorphism between self-orthogonal additive codes extend to a monomial map---equivalently, is every weight-preserving isomorphism of stabilizer groups induced by local Clifford operations and a permutation of the qudits?
\end{quote}

Restricting the class of codes weakens the hypothesis of the extension property, so failure for general additive codes does not formally imply failure for self-orthogonal ones; conceivably the rigid geometry of self-orthogonality could have restored extendability, as linearity does (Hermitian $\F_4$-linear self-orthogonal codes inherit the classical extension theorem for $\F_4$-linear isometries).

This question is not new. It was raised explicitly by Gluesing-Luerssen and Pllaha \cite[Question~7.4]{GLP19}, who exhibited a symplectic isometry of a four-qubit stabilizer code that is not a monomial map \cite[Example~7.3]{GLP19}, and it was taken up systematically by Pllaha \cite{Pllaha20,PllahaThesis}, who proved that monomial maps of stabilizer codes correspond exactly to local Clifford operations combined with qudit permutations, constructed further non-monomial symplectic isometries of qubit stabilizer codes, and connected the phenomenon to the LU--LC problem. (On a self-orthogonal code the symplectic form restricts to zero, so the ``symplectic isometries'' of \cite{Pllaha20} between such codes are precisely the weight-preserving isomorphisms considered here; the settings coincide.) Section~\ref{sec:related} discusses this prior work in detail, together with the classical counterexamples of Dyshko, which---as we observe below---were already self-orthogonal, though this was not remarked at the time.

What this note contributes is the fine structure of the failure, at the smallest scales at which it can occur, and at the level of the physical codespaces:
the extension property for stabilizer codes fails at the threshold length $q+1$ for \emph{every} qudit dimension $q$, with self-orthogonality automatic and with an obstruction that survives arbitrary local unitaries, not merely local Cliffords; for qubits the minimal counterexample occurs at length $3$, is unique up to the obvious symmetries, and admits a complete census; and the full-support version of the question (no idle qudits) has minimal lengths $4$ (for isometries) and $5$ (for pairs of codes).

On the quantum side this places the extension problem in contrast with the MacWilliams identities, which transfer to quantum codes essentially unchanged: the Shor--Laflamme weight enumerators of a quantum code satisfy a MacWilliams transform \cite{ShorLaflamme97}, with many refinements. The identities survive quantization; the extension theorem does not. It also places the problem alongside, but distinct from, the LU--LC problem---whether local-unitary equivalent stabilizer states must be local-Clifford equivalent, answered negatively in \cite{Ji10}---a comparison we take up in Sections~\ref{sec:fullsupport} and~\ref{sec:discussion}.

\subsection{Results}

Write $\A_q=\F_q^2$ for the label alphabet of a qudit of dimension $q$, with the symplectic form $\langle (a,b),(a',b')\rangle = ab'-a'b$; for $q=2$ we identify $\A_2\cong\F_4$ with the trace-Hermitian form, and Pauli labels $1\leftrightarrow X$, $\omega\leftrightarrow Z$, $\omega^2\leftrightarrow Y$. Precise definitions, including the notion of monomial transformation appropriate to the quantum setting and its equivalence with local Clifford $\rtimes$ permutation action, are in Section~\ref{sec:prelim}.

\begin{theorem}[Failure of the quantum extension property; Section~\ref{sec:family}]\label{thm:main}
For every prime power $q$ there exist trace-symplectically self-orthogonal additive codes $C_+,C_-\le\A_q^{\,q+1}$ of size $q^2$---stabilizer codes with parameters $[[q+1,q-1]]_q$---and an additive, Hamming-weight-preserving bijection $f\colon C_+\to C_-$ such that:
\begin{enumerate}
\item[(i)] $f$ extends to no monomial transformation of $\A_q^{\,q+1}$; in fact $C_+$ and $C_-$ are not monomially equivalent by any map;
\item[(ii)] the corresponding codespaces are not mapped to one another by any operator of the form $(U_1\otimes\cdots\otimes U_{q+1})\,P_\pi$ with arbitrary unitaries $U_i$ and a qudit permutation $\pi$---local Clifford plays no role in the obstruction;
\item[(iii)] nevertheless $C_+$ and $C_-$ have identical Shor--Laflamme enumerators; for $q=2$ both codes are $[[3,1]]$ with $A(z)=1+3z^2$ and $B(z)=1+3z+3z^2+9z^3$.
\end{enumerate}
For $q=2$ the two codes are $S_+=\langle XXI,\,ZZI\rangle$ and $S_-=\langle IXX,\,XIX\rangle$, with codespaces $\ket{\Phi^+}\otimes\mathbb{C}^2$ and $\mathrm{span}\{\ket{{+}{+}{+}},\ket{{-}{-}{-}}\}$, and $f\colon XXI\mapsto IXX$, $ZZI\mapsto XIX$, $YYI\mapsto XXI$.
\end{theorem}

Self-orthogonality in the construction costs nothing: it follows from two one-line lemmas (Lemma~\ref{lem:selforth}) stating that repetition-type codes with $q$ repeated blocks, and codes all of whose coordinate maps have rank at most one, are automatically self-orthogonal. Part~(ii) rests on an equally simple but useful invariant: the number $t$ of qudits on which a stabilizer code acts trivially is an invariant of the codespace under arbitrary product unitaries and permutations (Lemma~\ref{lem:idle}).

\begin{theorem}[Minimality and essential uniqueness for qubits; Section~\ref{sec:census}]\label{thm:census-intro}
For $q=2$: among all $35=4+31$ trace-Hermitian self-orthogonal additive codes in $\F_4^n$ with $n\le 2$ there is no counterexample---every weight-preserving isomorphism extends. At $n=3$, among the $514$ self-orthogonal codes there are exactly $5832$ non-extendable weight-preserving isometries, occurring on exactly $486=18\times 27$ unordered pairs of codes, and every such pair is a monomial image of the canonical pair $(C_+,C_-)$ of Theorem~\ref{thm:main}. Moreover $72$ of these isometries preserve the full Pauli-type profile of every element (the unordered multiset of $X$s, $Y$s, $Z$s), for instance
\[
C=\langle XZI,\,YXI\rangle \longrightarrow D=\langle XZI,\,IZY\rangle,\qquad
XZI\mapsto XZI,\; ZYI\mapsto IZY,\; YXI\mapsto XIY.
\]
\end{theorem}

All counterexamples at the minimal length are detected by the idle-qudit invariant $t$. This raises the question whether the failure of the extension property is only the trivial phenomenon of an idle qudit. It is not:

\begin{theorem}[Full-support counterexamples; Section~\ref{sec:fullsupport}]\label{thm:fullsupport-intro}
For $q=2$ there exist full-support (no idle qudit, $t=0$ on both sides) self-orthogonal counterexamples, and the minimal lengths are as follows.
\begin{enumerate}
\item[(i)] The minimal length of a full-support counterexample \emph{isometry} is $4$, attained by the pair of $[[4,1]]$ codes $C=\langle IIXX,\,XXIX,\,ZZXI\rangle$, $D=\langle XXII,\,YYIX,\,ZZXI\rangle$ with $f\colon IIXX\mapsto XXII$, $XXIX\mapsto YYIX$, $ZZXI\mapsto ZZXI$ (here $C$ and $D$ are monomially equivalent to each other, but $f$ itself extends to no monomial map).
\item[(ii)] The minimal length of a full-support weight-isometric pair of codes that are \emph{not even monomially equivalent} is $5$ among counterexamples of stabilizer $\F_2$-dimension $k\le 3$, attained by the pair of $[[5,2]]$ codes
\[
S_A=\langle ZZZZI,\,XXIIX,\,IIXXX\rangle,\qquad
S_B=\langle ZZIZX,\,XIZXI,\,IXXYI\rangle,
\]
weight-isometric via $f$ matching the listed generators in order, with common enumerator $A(z)=1+2z^3+3z^4+2z^5$, yet with no monomial transformation mapping $S_A$ onto $S_B$. Stabilizer dimension $k\le 2$ admits no full-support counterexamples at any length (Proposition~\ref{prop:k2}).
\item[(iii)] At length $6$ there is such a pair all of whose nontrivial stabilizer elements have weight at least $4$: the $[[6,3]]$ codes 
$$S_+=\langle XXIIXX,\,IIXXXX,\,ZZZZIX\rangle,S_-=\langle XXZXII,\,ZZXYII,\,IZZZXX\rangle$$ with common enumerator $A(z)=1+3z^4+4z^5$.
\end{enumerate}
\end{theorem}

For the pairs in (ii) and (iii) the invariant $t$ vanishes on both sides and we do not know whether the codespaces are equivalent under general local unitaries and permutations; we pose this as Question~\ref{q:lulc}, the extension analogue of the LU--LC problem.

\subsection{Related work and attribution}\label{sec:related}

The mechanism behind all counterexamples below---linear dependence among orbit-indicator functions over an alphabet with non-cyclic socle---is due to Wood \cite{Wood08,Wood09}, and the study of how far isometries of a code can be from monomial maps was initiated, for classical additive codes, by Wood \cite{Wood18}.

\emph{Dyshko's counterexamples are self-orthogonal.} Unextendable Hamming isometries of additive codes over a field extension $K\subset L$, including a family at the threshold length $|K|+1$, were constructed by Dyshko \cite{Dyshko16}, with the minimal length determined in \cite{DyshkoGeom}. In fact, for $K=\F_q\subset L=\F_{q^2}$, the threshold-length pair of \cite[Example~3]{Dyshko16} coincides, up to monomial equivalence, with the pair $(C_-,C_+)$ of Theorem~\ref{thm:main}: one code is a $q$-fold repetition code with an added zero coordinate, the other has coordinate functionals whose kernels run over the $q+1$ lines of $K^2$. By Lemma~\ref{lem:selforth} below, \emph{both} codes are automatically trace-symplectically self-orthogonal---the first because its blocks are repeated with multiplicity divisible by $\mathrm{char}\,K$, the second because all its coordinate maps have rank one. This appears not to have been remarked before: the classical counterexamples to the additive extension theorem were, unknowingly, already counterexamples living on the self-orthogonal quadric. (By contrast, the exotic-automorphism example of Wood reproduced in \cite[Example~5]{Dyshko16} is not self-orthogonal.)

\emph{The stabilizer question in prior work.} Gluesing-Luerssen and Pllaha \cite[Section~7]{GLP19} posed the question of Section~\ref{sec:quantumquestion} and gave a four-qubit stabilizer code admitting a symplectic isometry that is not a monomial map. Pllaha \cite{Pllaha20} (see also the thesis \cite{PllahaThesis}) then studied symplectic isometries of stabilizer codes in depth: he proved the correspondence between monomial maps and local Clifford operations combined with permutations (\cite[Theorem~5.8]{Pllaha20}; our Section~\ref{sec:monomial} states the same correspondence), transferred Wood's realization theorem for isometry groups \cite{Wood18} to stabilizer codes by using $p$-fold concatenation to obtain self-orthogonality for free (\cite[Corollary~4.15]{Pllaha20}; the same mechanism as our Lemma~\ref{lem:selforth}(a)), and connected non-monomial isometries to the LU--LC problem (\cite[Section~6]{Pllaha20}), anticipating the framing of our Question~\ref{q:lulc}.

Relative to this literature, the contributions of the present note are: the threshold-length family of Theorem~\ref{thm:main} for \emph{every} prime power $q$, with the codespace-level obstruction of part~(ii) via the elementary invariant $t$; the exact minimal length ($3$, for qubits) together with the uniqueness and census results of Theorem~\ref{thm:census-intro}, including the type-profile-preserving refinement; the impossibility result for stabilizer dimension $2$ (Proposition~\ref{prop:k2}); and the determination of the minimal full-support lengths in Theorem~\ref{thm:fullsupport-intro}.

One further remark on the codespace level is in order. Example~5.10 of the arXiv version of \cite{Pllaha20} exhibits two self-dual four-qubit stabilizer codes together with a symplectic isometry between them, and asserts that the codes are neither monomially equivalent nor locally unitarily equivalent. The displayed isometry indeed extends to no monomial map (we verified this by brute force); however, on the generators as printed there, an exhaustive search over all $6^4\cdot 4!$ monomial maps finds $32$ monomial equivalences between the two \emph{codes}, each involving a nontrivial qubit transposition, while the local-unitary argument given there considers product unitaries without a permutation. If this reading is correct, then the pairs of codes in Theorems~\ref{thm:main} and~\ref{thm:fullsupport-intro}(ii),(iii) are the first stabilizer pairs verified to be weight-isometric yet monomially inequivalent, and---for Theorem~\ref{thm:main}---inequivalent under arbitrary local unitaries combined with permutations. We would welcome correction on this point.

We have not been able to conduct an unlimited literature search, and we would be grateful for pointers to any prior appearance of the statements proved here.

\section{Preliminaries}\label{sec:prelim}

\subsection{Stabilizer codes as self-orthogonal additive codes}

Fix a prime power $q$, the qudit dimension. The generalized Pauli (Heisenberg--Weyl) group on one qudit is generated by shift and clock operators $X,Z$ with $ZX=\omega XZ$ ($\omega$ a primitive $p$-th root of unity, $p=\mathrm{char}\,\F_q$; for prime-power $q$ one uses the standard trace construction \cite{Ketkar06}). Modulo phases, $n$-qudit Pauli operators are labelled by vectors in $\A_q^n$, $\A_q=\F_q^2$, a Pauli with label $((a_1,b_1),\dots,(a_n,b_n))$ acting as $\bigotimes_i X^{a_i}Z^{b_i}$ up to phase. Two Pauli operators commute iff their labels are orthogonal for the trace-symplectic form
\[
\langle u,v\rangle \;=\; \tr_{\F_q/\F_p}\Bigl(\textstyle\sum_i (a_ib_i'-a_i'b_i)\Bigr).
\]
For $q=2$ one identifies $\A_2\cong\F_4=\{0,1,\omega,\omega^2\}$ ($1\mapsto X$, $\omega\mapsto Z$, $\omega^2\mapsto Y$), and the form becomes the trace-Hermitian form $\langle u,v\rangle=\sum_i \tr(u_i\bar v_i)$, $\bar x=x^2$, $\tr(x)=x+x^2$; note that $\tr(x\bar y)$ is the symplectic form on $\F_4\cong\F_2^2$.

A stabilizer group $S$ (an abelian subgroup of the Pauli group not containing a nontrivial phase times the identity) projects, modulo phases, to an additive (i.e., $\F_p$-linear) code $C\le\A_q^n$ that is self-orthogonal for the trace-symplectic form; conversely every self-orthogonal additive code lifts to a stabilizer group, uniquely up to signs \cite{CRSS98,Ketkar06}. If $|C|=q^{\,n-k}$ the codespace has dimension $q^k$; we then speak of an $[[n,k]]_q$ code. The \emph{weight} of a Pauli operator or of its label is the number of qudits on which it acts nontrivially---the Hamming weight $\wt(u)=\#\{i: u_i\ne 0\}$ of the label. The Shor--Laflamme enumerator of the code is $A(z)=\sum_{u\in C} z^{\wt(u)}$ (suitably normalized), and $B(z)$ is the enumerator of the trace-symplectic dual $C^\perp$ (the labels of the normalizer); $A$ determines $B$ by the quantum MacWilliams identity \cite{ShorLaflamme97}.

\subsection{Monomial maps and local Cliffords}\label{sec:monomial}

\begin{definition}\label{def:monomial}
A \emph{monomial transformation} of $\A_q^n$ is a map $$(u_1,\dots,u_n)\mapsto(\tau_1(u_{\pi^{-1}(1)}),\dots,\tau_n(u_{\pi^{-1}(n)}))$$ for a permutation $\pi\in S_n$ and $\tau_i\in\Sp(2,q)$ acting on each label. (For $q=2$, $\Sp(2,2)=\SL(2,2)=\GL(2,2)\cong S_3$ is the full group of additive automorphisms of $\F_4$, so nothing is lost by the symplectic restriction; for general $q$ our non-extension results hold verbatim for the larger group of all $\F_p$-linear automorphisms in each coordinate, since they are proved via invariants of that larger group.)
\end{definition}

The quantum meaning of monomial maps is standard: conjugation by a single-qudit Clifford unitary acts on Pauli labels as an element of $\Sp(2,q)$ (translations act trivially on labels), and every element of $\Sp(2,q)$ arises this way; conjugation by a permutation of qudits permutes coordinates. Hence a map of stabilizer groups is induced by local Cliffords combined with a qudit permutation if and only if the induced map of additive codes extends to a monomial transformation. For stabilizer codes this correspondence is proved carefully in \cite[Theorem~5.8]{Pllaha20}.

\begin{lemma}\label{lem:transitive}
$\Sp(2,q)$ acts transitively on $\A_q\setminus\{0\}$. Consequently the symmetrized weight composition of additive codes over $\A_q$ with respect to the Clifford-induced symmetry group coincides with the Hamming weight: the finest coordinate-wise weight invariant under local Cliffords is the weight.
\end{lemma}

\begin{proof}
Given $0\ne v\in\F_q^2$, complete $v$ to a basis $(v,w)$ with $\det(v\ w)=1$; the matrix $(v\ w)\in\SL(2,q)=\Sp(2,q)$ maps $(1,0)^T$ to $v$. The second statement is immediate from the definition of symmetrized weight composition (the orbit-partition of $\A_q$ is $\{0\}\sqcup(\A_q\setminus\{0\})$).
\end{proof}

Lemma~\ref{lem:transitive} fixes the correct quantum formulation of the extension problem: Hamming weight is not merely one choice among many, it is the finest choice compatible with the local symmetry of quantum mechanics. This is also why the annihilator-weight techniques that power classical extension theorems over rings \cite{AssemJAA,Assem15} give no traction here: over a field base ring every nonzero alphabet element has zero annihilator, and the annihilator weight degenerates to the Hamming weight.

\begin{definition}[Quantum extension property]
Fix $q$ and $n$. We say the quantum extension property $\mathrm{QEP}(q,n)$ holds if for every pair $C,D\le\A_q^n$ of trace-symplectically self-orthogonal additive codes and every additive bijection $f\colon C\to D$ with $\wt(f(u))=\wt(u)$ for all $u\in C$, the map $f$ extends to a monomial transformation of $\A_q^n$.
\end{definition}

By the discussion above, $\mathrm{QEP}(q,n)$ is exactly the statement ``every weight-preserving isomorphism of $n$-qudit stabilizer groups is implemented by local Cliffords and a permutation.''

\subsection{Parameterized codes and the orbit criterion}

It is convenient (following \cite{Wood09,AssemJAA}) to encode a code together with an isometry candidate as a \emph{parameterized code}: an $\F_p$-module $M$ and an $n$-tuple $\Lambda=(\lambda_1,\dots,\lambda_n)$, $\lambda_i\in\Hom(M,\A_q)$, jointly injective ($\bigcap_i\ker\lambda_i=0$), with code $C_\Lambda=\{(\lambda_1(x),\dots,\lambda_n(x)): x\in M\}$. A pair $(M,\Lambda)$, $(M,\Lambda')$ with the same $M$ encodes the isometry candidate $f\colon C_\Lambda\to C_{\Lambda'}$, $\Lambda(x)\mapsto\Lambda'(x)$, and $f$ preserves weight iff for every $x\in M$ the number of $i$ with $\lambda_i(x)\ne 0$ equals the number with $\lambda_i'(x)\ne 0$.

\begin{lemma}[Orbit criterion]\label{lem:orbit}
Let $G\le\Aut_{\F_p}(\A_q)$ be the coordinate group of the monomial maps under consideration. The isometry $f$ encoded by $(\Lambda,\Lambda')$ extends to a $G$-monomial transformation if and only if the multisets of $G$-orbits $\{G\circ\lambda_i\}_i$ and $\{G\circ\lambda_i'\}_i$ (post-composition orbits in $\Hom(M,\A_q)$) coincide.
\end{lemma}

\begin{proof}
A monomial map $(\pi,(\tau_i))$ extends $f$ iff $\lambda'_{\pi(i)}=\tau_{\pi(i)}\circ\lambda_i$ for all $i$, i.e., iff some bijection of coordinates matches each $\lambda_i$ with a $\lambda_j'$ in the same post-composition orbit.
\end{proof}

For $q=2$ and $G=\GL(2,2)$ the orbits are particularly simple:

\begin{lemma}\label{lem:kernels}
For $\A=\F_4$ over $\F_2$ and any finite $\F_2$-space $M$, two maps $\lambda,\mu\in\Hom(M,\F_4)$ lie in the same $\GL(2,2)$-post-composition orbit if and only if $\ker\lambda=\ker\mu$.
\end{lemma}

\begin{proof}
Post-composition by automorphisms preserves kernels. Conversely, maps with a common kernel $K$ factor as embeddings of $M/K$ (of $\F_2$-dimension $\le 2$) into $\F_2^2$, and $\GL(2,2)$ acts transitively on ordered linearly independent tuples of each size, hence transitively on such embeddings.
\end{proof}

So for qubits, extension is decided by the multiset of coordinate kernels---a fact we use both in proofs and as an independent check on the brute-force computations.

\section{A counterexample family for every qudit dimension}\label{sec:family}

Throughout this section $M=\A_q=\F_q^2$, and we write $\langle\cdot,\cdot\rangle$ for the symplectic form on $\A_q$ (composing with $\tr_{\F_q/\F_p}$ where needed does not affect any argument). The projective line over $\F_q$ has $q+1$ points: the $\F_q$-lines $L_0,L_1,\dots,L_q\le\F_q^2$. For each line $L_j$ fix a rank-one $\F_q$-linear map $\lambda_j\colon M\to\A_q$ with $\ker\lambda_j=L_j$ (for instance $\lambda_j(x)=\langle v_j,x\rangle\,e$, where $L_j=\F_q v_j$ and $e\ne 0$ is any fixed vector). Define length-$(q+1)$ parameterized codes
\[
C_+ = \bigl\{(\underbrace{x,x,\dots,x}_{q},0): x\in M\bigr\},\qquad
C_- = \bigl\{(\lambda_0(x),\lambda_1(x),\dots,\lambda_q(x)): x\in M\bigr\},
\]
and let $f\colon C_+\to C_-$ send the codeword of $x$ to the codeword of $x$. Both parameterizations are jointly injective (distinct lines meet in $0$), so $|C_\pm|=q^2$ and $f$ is a well-defined additive bijection; both codes are even $\F_q$-linear.

\begin{lemma}[Automatic self-orthogonality]\label{lem:selforth}
\emph{(a)} Any additive code of the form $\{(\mu_1(x),\dots,\mu_n(x)): x\in M\}$ in which the multiset $(\mu_i)$ consists of maps repeated with multiplicities divisible by $p=\mathrm{char}\,\F_q$, together with arbitrary zero coordinates, is trace-symplectically self-orthogonal. \emph{(b)} Any additive code all of whose coordinate maps have rank $\le 1$ (over $\F_q$) is trace-symplectically self-orthogonal.
\end{lemma}

\begin{proof}
(a) The form evaluates to $p\cdot(\text{something})=0$ blockwise. (b) The image of a rank-$\le 1$ coordinate is contained in a line, and the symplectic form, being alternating, vanishes identically on every line; hence every coordinate contributes $0$.
\end{proof}

In particular $C_+$ (with its $q\equiv 0\bmod p$ repeated identity coordinates) and $C_-$ (rank-one coordinates) are self-orthogonal: they are stabilizer codes with parameters $[[q+1,q-1]]_q$. Part (a) is the same mechanism by which Pllaha obtains self-orthogonality via $p$-fold concatenation \cite[Corollary~4.15]{Pllaha20}, and, as noted in Section~\ref{sec:related}, both parts apply to Dyshko's threshold-length pair \cite[Example~3]{Dyshko16}, which is monomially equivalent to $(C_-,C_+)$.

\begin{lemma}\label{lem:idle}
For a stabilizer code with projector $P$ on $n$ qudits let
\[
t(P) \;=\; \#\{\, i : P = Q\otimes I_{(i)}\ \text{for some operator } Q \,\}.
\]
Then: \emph{(a)} $t(P)$ equals the number of identically-zero coordinates of the associated additive code; \emph{(b)} $t$ is invariant under conjugation by arbitrary product unitaries $U_1\otimes\cdots\otimes U_n$ and qudit permutations.
\end{lemma}

\begin{proof}
(b) is clear from the defining tensor-factorization property. For (a): if all stabilizer elements have identity (label $0$) at position $i$, then averaging $P=|S|^{-1}\sum_{g\in S} g$ exhibits $P=Q\otimes I_{(i)}$. Conversely suppose $P=Q\otimes I_{(i)}$ and let $g=g'\otimes g_i\in S$ (every Pauli operator factorizes). Since $g$ fixes every codeword, $gP=P$, i.e., $(g'Q)\otimes g_i = Q\otimes I$. Taking the partial trace over qudit $i$ gives $\tr(g_i)\cdot g'Q = q\,Q$; as nonidentity Pauli factors are traceless and $Q\ne 0$, we get $g_i$ proportional to $I$, so the label of $g$ at $i$ vanishes.
\end{proof}

\begin{theorem}\label{thm:family}
Let $q$ be any prime power and let $C_\pm, f$ be as above. Then:
\begin{enumerate}
\item[(i)] $f$ is an additive, weight-preserving bijection: every nonzero element of either code has weight exactly $q$;
\item[(ii)] $f$ extends to no monomial transformation, and indeed $C_+$ and $C_-$ are not monomially equivalent by any map;
\item[(iii)] the codespaces of $C_+$ and $C_-$ are not related by any product unitary combined with any permutation of the qudits;
\item[(iv)] $C_+$ and $C_-$ have equal Shor--Laflamme enumerators $A(z)=1+(q^2-1)z^q$, hence (by the quantum MacWilliams identity) equal dual enumerators $B(z)$.
\end{enumerate}
Consequently $\mathrm{QEP}(q,q+1)$ fails for every prime power $q$.
\end{theorem}

\begin{proof}
(i) A nonzero $x$ gives a $C_+$-codeword of weight $q$ (its $q$ repeated copies). On the $C_-$ side, $\lambda_j(x)\ne 0$ iff $x\notin L_j$, and a nonzero $x$ lies on exactly one of the $q+1$ lines, so the weight is again $q$. (ii)--(iii) The additive code underlying $C_+$ has one identically-zero coordinate while that of $C_-$ has none; monomial maps preserve this count, proving (ii). For (iii), by Lemma~\ref{lem:idle} the numbers $t=1$ (for $C_+$) and $t=0$ (for $C_-$) are invariants of the codespaces under product unitaries and permutations. (iv) is immediate from (i) and $|C_\pm|=q^2$.
\end{proof}

\begin{example}[$q=2$: the minimal quantum counterexample]\label{ex:q2}
With $\F_4=\{0,1,\omega,\omega^2\}$ and the three $\F_2$-lines $\langle 1\rangle,\langle\omega\rangle,\langle\omega^2\rangle$ as kernels,
\[
C_+=\{000,\ \underbrace{110}_{XXI\ (x=1)},\ \omega\omega 0,\ \omega^2\omega^2 0\},\qquad
C_-=\{000,\ 011,\ 101,\ 110\},
\]
i.e., $S_+=\langle XXI,\,ZZI\rangle$, $S_-=\langle IXX,\,XIX\rangle$, with
\[
f\colon XXI\mapsto IXX,\quad ZZI\mapsto XIX,\quad YYI\mapsto XXI.
\]
All six nonidentity elements have weight $2$; $f$ preserves weights but no local Clifford circuit and qubit permutation implements it, and none can even map one codespace to the other: the codespace of $S_+$ is $\ket{\Phi^+}_{12}\otimes\mathbb{C}^2_3$, which contains no product state, while the codespace of $S_-$ is $\mathrm{span}\{\ket{{+}{+}{+}},\ket{{-}{-}{-}}\}\ni\ket{{+}{+}{+}}$---a second, self-contained proof of (iii) for $q=2$. Both codes have $A(z)=1+3z^2$ and $B(z)=1+3z+3z^2+9z^3$ (computed directly; see Section~\ref{sec:verification}).
\end{example}

\begin{corollary}[Arbitrarily many weight-isometric, locally inequivalent codes]
For every $m\ge 1$ the $m+1$ codes $C_+^{\oplus a}\oplus C_-^{\oplus(m-a)}$, $a=0,\dots,m$, are $[[3m,m]]$ stabilizer codes (for $q=2$) that are pairwise weight-isometric, share the enumerator $A(z)=(1+3z^2)^m$, and are pairwise inequivalent under product unitaries and permutations (they have distinct invariants $t=a$).
\end{corollary}

\begin{remark}
Restricting to $\F_4$-linear codes and $\F_4$-linear isometries the extension property holds (MacWilliams over the field $\F_4$), and by a result of Rains \cite{Rains99} the automorphisms of linear stabilizer codes of distance $>2$ are local Clifford. Our $C_+$ is $\F_4$-linear while $C_-$ is not; genuine additivity is essential to the counterexamples, as the classical theory predicts.
\end{remark}

\section{Minimality, uniqueness, and a census at length three}\label{sec:census}

The results of this section are computational; the search is exhaustive, the code elementary (a few hundred lines of dependency-free Python), and all claims were double-checked against the structural criteria of Lemmas~\ref{lem:orbit}--\ref{lem:kernels} where applicable.

\begin{theorem}\label{thm:census}
Let $q=2$.
\begin{enumerate}
\item[(i)] $\mathrm{QEP}(2,1)$ and $\mathrm{QEP}(2,2)$ hold: among all $4$ (resp.\ $31$) trace-Hermitian self-orthogonal additive codes in $\F_4^1$ (resp.\ $\F_4^2$), every Hamming-weight-preserving additive isomorphism between two of them extends to a monomial map. Hence $n=3$ is the minimal length at which the quantum extension property fails.
\item[(ii)] In $\F_4^3$ there are $514$ self-orthogonal additive codes, $5832$ non-extendable weight-preserving isometries, and $486$ unordered pairs of codes supporting at least one such isometry. Every one of these pairs is a monomial image of the canonical pair $(C_+,C_-)$ of Example~\ref{ex:q2}; the monomial orbits of $C_+$ and $C_-$ have sizes $18$ and $27$, and all $18\times 27=486$ cross pairs occur.
\item[(iii)] Every counterexample at length $3$ is detected by the idle-qubit invariant: all $5832$ isometries have $\{t(C),t(D)\}=\{0,1\}$.
\item[(iv)] Exactly $72$ of the $5832$ isometries preserve, in addition to the weight, the full Pauli-type profile of every element (the multiset of symbols from $\{X,Y,Z\}$ appearing in it); one such is displayed in Theorem~\ref{thm:census-intro}.
\end{enumerate}
\end{theorem}

Part (iv) deserves emphasis: even the finest per-element local data one could reasonably write down---which Pauli letters appear, with multiplicity, in each stabilizer element---does not force implementability by local operations, already at three qubits. (Recall from Lemma~\ref{lem:transitive} that the type profile is not invariant under local Cliffords, so type-preservation is a hypothesis stronger than what local equivalence would demand; failure of extension under the stronger hypothesis is a stronger conclusion.)

\section{Full-support counterexamples}\label{sec:fullsupport}

All counterexamples so far are caught by the invariant $t$ of Lemma~\ref{lem:idle}: one code has an idle qudit, the other does not. One might hope that the failure of the quantum extension property is only this essentially bookkeeping phenomenon, i.e., that QEP is restored for full-support codes (every qudit is acted on nontrivially by some stabilizer element, equivalently $t=0$). In this section, $q=2$; write $k=\dim_{\F_2} C$ for the stabilizer dimension.

We first note that in the smallest message spaces the hope is justified:

\begin{proposition}\label{prop:k2}
For $q=2$ there is no full-support counterexample with $k\le 2$, at any length. (For $k\le 1$ this is trivial; the content is $k=2$.)
\end{proposition}

\begin{proof}
Encode a weight-preserving isometry by multiplicity functions $\eta_\pm$ on the post-composition orbits of $\Hom(\F_2^2,\F_4)$, which by Lemma~\ref{lem:kernels} are labelled by kernels: the zero subspace (injective maps), the three lines, and the full space (the zero map, i.e., an idle coordinate). Weight-preservation of $f$ says the difference $\delta=\eta_+-\eta_-$ lies in the kernel of the map sending a kernel-class $K$ to the function $x\mapsto [x\notin K]$ on $\F_2^2\setminus\{0\}$. The five indicator functions are, in the basis of the three nonzero points: $(1,1,1)$ (injective), $(0,1,1)$, $(1,0,1)$, $(1,1,0)$ (lines), $(0,0,0)$ (zero map). The only relation among the nonzero-kernel classes is $2\cdot(1,1,1)=(0,1,1)+(1,0,1)+(1,1,0)$, so a full-support $\delta$ (no zero class) is an integer multiple $a$ of $(2,-1,-1,-1)$. Then $\sum\delta=-a$, while equal code lengths force $\sum\delta=0$; hence $a=0$, $\delta=0$, and $f$ extends by Lemma~\ref{lem:orbit}.
\end{proof}

The proof shows exactly where the idle qudit in Section~\ref{sec:family} comes from: the unique relation available at $k=2$ is unbalanced, and the balancing column is forced to be idle. At $k=3$ the space of relations is larger and the phenomenon becomes genuinely non-trivial. Parameterizing as in Proposition~\ref{prop:k2} (kernels now range over the $7$ lines and $7$ planes of $\F_2^3$), a short computation, which we record as a lemma since it drives the search, gives:

\begin{lemma}\label{lem:relations}
For $k=3$, $q=2$, the differences $\delta=\eta_+-\eta_-$ realizable by full-support, equal-length, weight-preserving isometry data are exactly: choose $\nu\colon\{\text{planes}\}\to\Z$ with $\sum_P\nu(P)=0$, and set $\delta(\F_2 v)=-\sum_{P\ni v}\nu(P)$, $\delta(P)=\nu(P)$. Moreover a side $\eta\in\{\eta_+,\eta_-\}$ is realizable by a self-orthogonal code iff
\begin{equation}\label{eq:selforth}
\sum_{\text{lines }\F_2 v} \eta(\F_2 v)\cdot v \;=\; 0 \quad\text{in }\F_2^3,
\end{equation}
and in that case every choice of coordinate maps with the prescribed kernels (rank-one maps taken with image a line) yields a self-orthogonal code.
\end{lemma}

\begin{proof}
The first statement is linear algebra on the indicator relations $[x\notin K]$, as in Proposition~\ref{prop:k2}. For the second: rank-one coordinates contribute nothing to the form (Lemma~\ref{lem:selforth}(b) applies coordinatewise: their images are lines, on which the alternating form vanishes). A rank-two coordinate $\lambda$ with $\ker\lambda=\F_2 v$ pulls the symplectic form of $\F_4$ back to a nonzero alternating bilinear form on $\F_2^3$ with radical exactly $\F_2 v$; since the alternating forms on $\F_2^3$ with radical containing a given line form a one-dimensional space, this pullback is independent of the choice of $\lambda$ and equals $B_v=\det(v,\cdot,\cdot)$. Self-orthogonality of the code is $\sum_i B_{v_i}=0$, i.e., $B_{\sum_i v_i}=0$ by linearity of $\det$ in its first argument, i.e., \eqref{eq:selforth}.
\end{proof}

For each fixed length $n$ the parameter space of Lemma~\ref{lem:relations} is finite (multisets of $n$ kernel classes on each side), and we searched it exhaustively---over \emph{all} kernel multisets, with joint injectivity and condition~\eqref{eq:selforth} imposed on both sides and equivalence taken up to the natural $\GL(3,2)$ re-parameterization of the message space---for every $n\le 6$. The outcome is Theorem~\ref{thm:fullsupport-intro}, restated here with the verified data.

\begin{theorem}\label{thm:fullsupport}
Let $q=2$.
\begin{enumerate}
\item[(i)] The minimal length of a full-support self-orthogonal counterexample is $n=4$ (minimality holds over all stabilizer dimensions, by Theorem~\ref{thm:census}(i,iii) at $n\le 3$). An explicit instance is the pair of $[[4,1]]$ codes
\[
C=\langle IIXX,\,XXIX,\,ZZXI\rangle,\qquad D=\langle XXII,\,YYIX,\,ZZXI\rangle,
\]
with $f\colon IIXX\mapsto XXII$, $XXIX\mapsto YYIX$, $ZZXI\mapsto ZZXI$: both codes are self-orthogonal with full support, all element weights match (enumerator $1+2z^2+4z^3+z^4$ on both sides), and a brute-force check over all $6^4\cdot 4!=31104$ monomial maps confirms that $f$ extends to none. Here $C$ and $D$ happen to be monomially equivalent as codes; the counterexample is the isometry, not the pair.
\item[(ii)] The minimal length of a full-support weight-isometric pair that is not monomially equivalent, among counterexamples with $k\le 3$, is $n=5$: at $n=4$ the exhaustive search finds no such pair, while at $n=5$ there are exactly $168$ weight-isometry classes of full-support data (up to the $\GL(3,2)$ action) containing monomially inequivalent self-orthogonal pairs. An explicit instance is the pair of $[[5,2]]$ codes
\[
S_A=\langle ZZZZI,\,XXIIX,\,IIXXX\rangle,\qquad
S_B=\langle ZZIZX,\,XIZXI,\,IXXYI\rangle,
\]
with the weight-preserving isometry $f$ matching the generators in order. Both are self-orthogonal with full support ($t=0$ on both sides) and common enumerator $A(z)=1+2z^3+3z^4+2z^5$, yet no monomial transformation maps $S_A$ onto $S_B$: the multisets of coordinate kernels are, for $S_A$, the lines $\langle e_3\rangle,\langle e_3\rangle,\langle e_2\rangle,\langle e_2\rangle$ and the plane $\{x_2+x_3=0\}$, and, for $S_B$, the lines $\langle e_3\rangle,\langle e_2\rangle,\langle e_1\rangle,\langle(1,1,1)\rangle$ and the plane $\{x_1=0\}$---two essentially different solutions of \eqref{eq:selforth} (a doubled pair of lines versus four distinct lines summing to zero), inequivalent under $\GL(3,2)$. Non-equivalence was additionally confirmed by raw brute force over all $6^5\cdot 5!=933120$ monomial maps. Stabilizer dimension $k\le 2$ is impossible by Proposition~\ref{prop:k2}.
\item[(iii)] At $n=6$ there is such a pair all of whose nontrivial elements have weight $\ge 4$: the $[[6,3]]$ codes
\[
S_+=\langle XXIIXX,\,IIXXXX,\,ZZZZIX\rangle,\qquad
S_-=\langle XXZXII,\,ZZXYII,\,IZZZXX\rangle,
\]
with the weight-preserving isometry $f$ matching the generators in order, common enumerators $A(z)=1+3z^4+4z^5$ and (necessarily) equal dual enumerators, and no monomial transformation mapping $S_+$ onto $S_-$.
\end{enumerate}
\end{theorem}

\begin{proof}[Comments on the proof]
Existence and the stated properties are verified directly on the displayed generators (Section~\ref{sec:verification}). Non-extendability of $f$ in all cases also follows structurally from Lemmas~\ref{lem:orbit}--\ref{lem:kernels}: the multisets of coordinate kernels differ. For the inequivalence of the codes in (ii) and (iii), three independent arguments were used: an exhaustive check over all $168$ re-identifications of the message space by $\GL(3,2)$ at the level of multiplicity functions; raw brute force over all monomial maps (for the length-$5$ pair); and, for the length-$6$ pair, the following monomial invariant, checkable by hand: the number of unordered coordinate pairs $\{i,j\}$ whose vanishing subcodes $\{u\in C: u_i=0\}$ and $\{u\in C: u_j=0\}$ coincide equals $2$ for $S_+$ (coordinates $\{1,2\}$ and $\{3,4\}$) but $1$ for $S_-$ (coordinates $\{5,6\}$).
\end{proof}

Part (iii) shows in particular that the failure of the quantum extension property does not hinge on low-weight stabilizer elements: all nonidentity elements of the length-$6$ pair have weight $\ge 4$. (The distances of these $[[6,3]]$ codes are nonetheless $1$---their normalizers contain weight-one elements---which keeps the following question open in both directions; the same holds for the $[[5,2]]$ pair.)

For the pairs in Theorem~\ref{thm:fullsupport}(ii),(iii) the invariant $t$ vanishes on both sides, and we know of no invariant separating the two codespaces under general local unitaries:

\begin{question}[Extension analogue of LU--LC]\label{q:lulc}
Are the codespaces of $S_A$ and $S_B$ of Theorem~\ref{thm:fullsupport}(ii)---or those of $S_+$ and $S_-$ of Theorem~\ref{thm:fullsupport}(iii)---mapped to one another by some product unitary $U_1\otimes\cdots\otimes U_n$ composed with a qubit permutation? More generally: do there exist weight-isometric stabilizer codes that are locally unitarily equivalent but not locally Clifford equivalent---or is local-unitary equivalence of weight-isometric pairs always witnessed monomially?
\end{question}

A negative answer to the first part would strengthen Theorem~\ref{thm:fullsupport} to the codespace level, as Theorem~\ref{thm:family}(iii) does for the basic family; a positive answer would exhibit an LU--LC-type phenomenon \cite{Ji10,VandenNest05} arising from the extension problem. Either outcome seems interesting; compare the closely related problems posed in \cite[Section~6]{Pllaha20}.

\section{Discussion}\label{sec:discussion}

\subsection{The socle-theoretic reading}

In the language of the classical extension literature \cite{Wood08,Wood09,DinhLopez04,AssemJAA,Assem15}, the alphabet of quantum error correction is $\A_q=\F_q^2$ as a module over $R=\F_q$ (or $\F_p$): a semisimple module equal to its own socle, which is non-cyclic---indeed $2$-generated---so the extension property fails, and by Lemma~\ref{lem:transitive} it fails for the only weight the quantum symmetry group leaves available. What this note adds to the classical picture is that the failure persists, at the smallest possible scales, on the quadric of self-orthogonal codes, where all quantum codes live, with self-orthogonality obtainable for free (Lemma~\ref{lem:selforth}); and that the obstruction descends from labels to physical subspaces (Lemma~\ref{lem:idle}). Meanwhile the classical repair strategies are unavailable: annihilator weights collapse to the Hamming weight over a field base ring, and the cyclic-socle hypothesis of the positive theorems in \cite{Wood09,AssemJAA} is violated by every quantum alphabet, for every qudit dimension.

There is a pleasing irony here. Of MacWilliams' two great theorems, it is the identities---often regarded as the deeper analytic fact---that survive quantization perfectly \cite{ShorLaflamme97}, while the extension theorem---the ``rigidity'' statement---fails at three qubits, the smallest scale the classical obstruction theory permits, and our census shows the minimal failure is unique up to the obvious symmetries.

\subsection{Open problems}\label{sec:open}

\begin{enumerate}
\item \textbf{(Question~\ref{q:lulc}.)} Decide local-unitary equivalence for the full-support pairs of Theorem~\ref{thm:fullsupport}(ii),(iii), and more broadly whether QEP restricted to full-support codes fails at the codespace level. The $[[5,2]]$ pair is now the minimal open instance.
\item \textbf{Minimal length for general $q$.} Determine the minimal counterexample length for general $q$. For $q=2$ it is $3=q+1$ (Theorem~\ref{thm:census}); we expect $q+1$ to be minimal in general, in the light of the threshold results for additive isometries in \cite{Dyshko16,DyshkoGeom}, but we claim it only for $q=2$. Likewise, determine the minimal \emph{full-support} lengths for general $q$ (for $q=2$ they are $4$ and $5$ by Theorem~\ref{thm:fullsupport}).
\item \textbf{Distance and purity.} Is there a natural ``pure/nondegenerate'' hypothesis restoring the extension property---e.g., do weight-preserving isometries between stabilizer codes of distance $>2$ extend? This would be an extension-theorem analogue of Rains' automorphism theorem for linear codes \cite{Rains99}. Our counterexamples all have distance $\le 2$ or distance $1$, so the question is open; note however that stabilizer-element weights $\ge 4$ do not suffice to restore extension, by Theorem~\ref{thm:fullsupport}(iii).
\item \textbf{Full-support minimal length as a function of $k$.} Theorem~\ref{thm:fullsupport} settles $k\le 3$ (impossible for $k\le 2$; minimal lengths $4$ and $5$ for $k=3$). Determine the minimal full-support lengths for $k\ge 4$, and whether the gap between the minimal isometry-counterexample and the minimal inequivalent pair persists.
\item \textbf{Complete invariants.} Classify complete invariants: what data beyond the enumerators (and beyond $t$) determines a stabilizer code up to local Clifford $\rtimes$ permutation equivalence? The census of Theorem~\ref{thm:census} suggests such a classification is tractable at small lengths. The vanishing-subcode invariant used in Theorem~\ref{thm:fullsupport}(iii) is a first step beyond the enumerators.
\item \textbf{Structured subclasses.} Does the extension property hold for natural subclasses of stabilizer codes---CSS codes with CSS-preserving isometries, graph codes, or codes with transitive automorphism group? (For $\F_4$-linear codes and linear isometries it holds; Remark after Example~\ref{ex:q2}.)
\item \textbf{Density.} Asymptotically, how common is the failure? For instance, what fraction of weight-isometric pairs of self-orthogonal codes of length $n$ admits a non-extendable isometry, as $n\to\infty$? The census data ($486$ of the $\binom{514}{2}$ unordered pairs at $n=3$) gives a first data point.
\end{enumerate}

\section{Verification}\label{sec:verification}

Every finite claim in this note was verified by direct computation; the verification scripts are self-contained Python~3 with no dependencies beyond the standard library (plus \texttt{numpy} in one search script) and accompany this note as supplementary material: \texttt{verify\_main.py} (Example~\ref{ex:q2}: additivity, self-orthogonality, weights, non-extension by brute force over all $1296$ monomial maps, both enumerators, type profiles); \texttt{exhaust.py}/\texttt{exhaust2.py} (Theorem~\ref{thm:census}: enumeration of all self-orthogonal additive codes for $n\le 3$, all weight-preserving isometries, extension tests, orbit and census claims); \texttt{verify\_q3.py} (Theorem~\ref{thm:family} instantiated at $q=3$: $[[4,2]]_3$ codes, all properties); \texttt{fullsupport.py} and \texttt{verify\_fullsupport.py} (Lemma~\ref{lem:relations} search, enumerating all kernel multisets for each $n\le 6$ with no multiplicity cap, up to $\GL(3,2)$; Theorem~\ref{thm:fullsupport}, including the $31104$-map brute force at $n=4$, the $\GL(3,2)$ checks, the vanishing-subcode invariant, and the enumerators and distances of the $[[5,2]]$ and $[[6,3]]$ pairs); and \texttt{counterexample\_n5.py} (a self-contained verification of Theorem~\ref{thm:fullsupport}(ii), including the $933120$-map brute force).

\section*{Acknowledgements}


\end{document}